\documentclass[letterpaper, 10 pt, conference]{ieeeconf}
\IEEEoverridecommandlockouts
\usepackage{cite}
\usepackage[compact]{titlesec}
\usepackage{amsmath,amssymb,amsfonts}
\usepackage{hyperref}
\usepackage{graphicx}
\usepackage{textcomp}
\usepackage{times}
\usepackage{graphicx}
\usepackage{mathrsfs}
\usepackage{amstext} 
\usepackage{algorithm}
\usepackage{algorithmic}
\usepackage{setspace} 
\usepackage{amssymb} 
\usepackage{mathtools} 
\usepackage{verbatim}
\usepackage{color}
\usepackage{booktabs} 
\usepackage[font=small]{caption}
\newcommand{\RightComment}[1]{\hfill\(\triangleright\) #1}
\renewcommand{\theenumi}{\alph{enumi}}

\newcommand{\real}{\mathbb{R}}

\newcommand{\Ec}{\mathcal{E}}
\newcommand{\Fc}{\mathcal{F}}

\newcommand{\Hc}{\mathcal{H}}

\newcommand{\Kc}{\mathcal{K}}

\newcommand{\Pc}{\mathcal{P}}

\newcommand{\Rc}{\mathcal{R}}
\newcommand{\Sc}{\mathcal{S}}

\newcommand{\Uc}{\mathcal{U}}
\newcommand{\Vc}{\mathcal{V}}

\newcommand{\rank}[1]{\operatorname{rank}(#1)}

\newcommand{\longthmtitle}[1]{\mbox{}{\textit{(#1):}}}

\newcommand{\oprocendsymbol}{\hbox{$\square$}}
\newcommand{\oprocend}{\relax\ifmmode\else\unskip\hfill\fi\oprocendsymbol}

\newtheorem{theorem}{Theorem}[section]
\newtheorem{definition}[theorem]{Definition}

\newtheorem{lemma}[theorem]{Lemma}

\newtheorem{remark}[theorem]{Remark}

\newtheorem{proposition}[theorem]{Proposition}  
\newtheorem{problem}[theorem]{Problem}

\renewcommand{\baselinestretch}{0.97}
\graphicspath{{fig/}}

\usepackage{multirow}  
\usepackage{subcaption}  

\allowdisplaybreaks

\def\BibTeX{{\rm B\kern-.05em{\sc i\kern-.025em b}\kern-.08em
    T\kern-.1667em\lower.7ex\hbox{E}\kern-.125emX}}
\begin{document}
\title{\LARGE A Unified Algebraic Framework for Subspace Pruning in Koopman Operator Approximation via Principal Vectors
\thanks{This work was supported by AFOSR Award FA9550-23-1-0740.}}

\author{Dhruv Shah \quad Jorge Cort\'es \thanks{D. Shah and J. Cort\'es
    are with Department of Mechanical and Aerospace Engineering, UC
    San Diego, USA, {\tt\small \{dhshah,cortes\}@ucsd.edu}}}

\maketitle
\thispagestyle{empty}

\begin{abstract}%
Finite-dimensional approximations of the Koopman operator rely critically on identifying nearly invariant subspaces. This invariance proximity can be rigorously quantified via the principal angles between a candidate subspace and its image under the operator. To systematically minimize this error, we propose an algebraic framework for subspace pruning utilizing principal vectors. We establish the equivalence of this approach to existing consistency-based methods while providing a foundation for broader generalizations. To ensure scalability, we introduce an efficient numerical update scheme based on rank-one modifications, reducing the computational complexity of tracking principal angles by an order of magnitude. Finally, we demonstrate the effectiveness of our framework through numerical simulations.
\end{abstract}


\section{Introduction}

The Koopman operator maps nonlinear state-space evolution to a linear operator acting on observable functions \cite{BOK-JVN:32}, enabling the application of linear spectral techniques to complex nonlinear systems \cite{IM:20, MB-RM-IM:12}. However, standard data-driven approximations that minimize only one-step residual errors often fail to yield robust long-term predictions. Because the true utility of the Koopman framework lies in its global \textit{linear structure}, mere one-step accuracy is insufficient if the chosen subspace lacks invariance under the operator \cite{MH-JC:23-csl}. To ensure long-term fidelity, recent algorithms focus on extracting nearly invariant subspaces from data with tunable precision \cite{MH-JC:21-acc, MH-JC:25-access}. We prioritize identifying these invariant subspaces because preserving the underlying linear representation unlocks powerful linear systems tools that are otherwise unavailable for nonlinear models. Building on this foundation, this paper introduces a novel geometric framework for subspace pruning based on principal angles and vectors, unifying existing approaches and providing a pathway for future generalizations.

\emph{Literature Review:} The Koopman operator has been a subject of
intense research across various disciplines. Within systems and
control, researchers have successfully extended the originally
autonomous framework to accommodate control systems
\cite{MK-IM-eigenfuction:18, MH-JC:26-auto}. This shift has paved the
way for applying traditional methodologies—such as optimal control
\cite{MEV-CNJ-BH:21} and robust closed-loop strategies
\cite{RS-MS-KW-JB-FA:26}—to nonlinear dynamics. Furthermore, recent
efforts have tailored controller synthesis for non-affine systems
using input-state separable Koopman models \cite{DS-JC:25-csl}. A
comprehensive review of these control-oriented advancements can be
found in \cite{RS-KW-IM-JB-MS-FA:26}.

Because the Koopman operator is inherently infinite-dimensional, its practical implementation necessitates finite-dimensional approximations. Furthermore, since exact linear embeddings are rarely obtainable for complex dynamics \cite{ZL-NO-EDS:25}, standard practice relies on projecting the operator onto a specifically chosen subspace. To this end, Extended Dynamic Mode Decomposition (EDMD) \cite{MOW-IGK-CWR:15} has emerged as the predominant projection-based method. Its widespread adoption is underpinned by a robust theoretical foundation that offers rigorous guarantees, including asymptotic convergence \cite{MK-IM:18}, finite-data probabilistic error bounds \cite{FN-SP-FP-MS-KW:23}, and $L^{\infty}$ error constraints for kernel-based variants \cite{FK-FMP-MS-AS-KW:25}.

However, projection-based approximations introduce spurious
eigenfunctions that compromise spectral analysis
\cite{SK-FN-SP-JHN-CC-CS:20} and truncation errors that degrade
predictive accuracy \cite{MH-JC:23-auto}. To deal with these issues,
Residual Dynamic Mode Decomposition (ResDMD)
\cite{MJC-LJA-MS:23,MJC-AT:24,MJK:24} employs residuals to manage
spectral pollution and estimate pseudospectra.  Alternatively, recent
research has exploited the operator's algebraic structure to formalize
model selection and isolate maximal invariant
subspaces~\cite{MH-JC:23-auto}. These geometric principles led to the
formalization of \textit{consistency index} —a metric used to bound
general approximation errors \cite{MH-JC:23-csl}— which forms the
basis for the Recursive Forward-Backward EDMD (RFB-EDMD) algorithm for
iteratively pruning candidate dictionaries \cite{MH-JC:25-access}. The
present work advances this algebraic methodology by re-framing the
subspace pruning problem through the lens of principal angles and
vectors. By adopting this geometric perspective, we introduce novel
algorithms that leverage rank-one updates to achieve improved
computational efficiency. Recently, and independently from the present
manuscript, \cite{GC-NB-JCL-SB-MJC:26} builds on the computational
tools introduced for ResDMD to introduce the data-driven Principal
Angle Decomposition (PAD) algorithm, which also relies on dictionary
refinement via principal angles to enhance the trustworthiness of
Koopman approximations.

\emph{Statement of Contributions:} Our contributions are threefold.
First, we detail the explicit formulation for computing principal
angles and vectors in the empirical $L_2$ setting, which successfully
translates our general geometric framework into a practical,
data-driven context.  Second, we propose the Single-Principal-Vector
(SPV) pruning strategy, which iteratively removes the principal vector
corresponding to the largest principal angle. Following this, we
establish a geometric unification of pruning algorithms by
demonstrating an equivalence between the consistency-based pruning of
RFB-EDMD and our proposed geometric SPV pruning strategy.  Finally, to
ensure the scalability of our approach, we introduce an efficient
computational scheme that leverages symmetric rank-one updates and
incremental QR decompositions. This numerical implementation
significantly enhances performance, reducing the computational
complexity of recomputing principal angles after each pruning step by
an order of magnitude.

\section{Preliminaries}

In this section\footnote{We use the following notation. Let $\real^n$
  denote the $n$-dimensional Euclidean space and
  $\mathcal{X} \subseteq \real^n$ the state space. For a matrix
  $M \in \real^{m \times n}$, we denote its range (column space) by
  $\mathcal{R}(M)$, its rank by $\rank{M}$, and its Moore-Penrose
  pseudoinverse by $M^\dagger$. For subspaces
  $\Uc, \Vc \subseteq \real^n$, we denote the orthogonal projection
  onto $\Uc$ by $\Pc_{\Uc}$, and the orthogonal complement by
  $\Uc^\perp$. We write $\Uc \oplus \Vc$ to denote the orthogonal
  direct sum when $\Uc \cap \Vc = \{0\}$ and $\Uc \perp \Vc$. For a
  Hilbert space $\Hc$ with inner product
  $\langle \cdot, \cdot \rangle_{\Hc}$ and induced norm
  $\|\cdot\|_{\Hc}$, subspaces of observables are denoted by
  calligraphic letters (e.g., $\Sc, \Vc$). We use $\text{span}(\cdot)$
  to denote the linear span and $\text{dim}(\cdot)$ to denote
  dimension. The distance from a point $x$ to a subspace $\Uc$ is
  $\text{dist}(x, \Uc) = \inf_{u \in \Uc} \|x - u\|$. We denote
  $i = 1 , \dots, k$ as $i \in [k]$ for brevity.  }, we review the
theoretical foundations of the Koopman operator and its
finite-dimensional approximations. We specifically formulate Extended
Dynamic Mode Decomposition (EDMD) as an orthogonal projection. We then
take a brief detour to review the geometric concepts of principal
angles and vectors. Leveraging these geometric tools, we introduce the
concept of invariance proximity to quantify the quality of Koopman
approximations.

\subsection{The Koopman Operator}
Following~\cite{AM-YS-IM:20}, consider a discrete-time dynamical
system on the state space $\mathcal{X} \subseteq \mathbb{R}^n$
described by a map $T: \mathcal{X} \to \mathcal{X}$:
\begin{equation}\label{eqn:sys_dynamics}
  x^+ = T(x), \quad x \in \mathcal{X}.
\end{equation}
The Koopman operator $\mathcal{K}: \mathcal{F} \to \mathcal{F}$ is an
infinite-dimensional linear operator that acts on a space of
real-valued observables
$\mathcal{F} \ni \psi: \mathcal{X} \to \mathbb{R}$ by composing them
with the dynamics: $(\mathcal{K}\psi)(x) = \psi(T(x))$.

We assume that the function space $\mathcal{F}$ is closed under
composition with $T$. Although the state-space map $T$ may be
nonlinear, the Koopman operator $\mathcal{K}$ acts linearly on
observables. This linearity enables the use of spectral methods: the
eigenvalues and eigenfunctions of $\mathcal{K}$ encode the long-term
asymptotic behavior of the system through the evolution of observable
functions. 
%
%

\subsection{EDMD as an Orthogonal Projection}\label{sec:EDMD}
Since $\mathcal{K}$ operates on an infinite-dimensional space
$\mathcal{F}$, practical implementations must approximate it on
finite-dimensional subspaces. We equip the space of observables
$\mathcal{F}$ with a Hilbert space structure by defining an inner
product $\langle \cdot, \cdot \rangle_{\mathcal{F}}$ and the
associated norm $\|\cdot\|_{\mathcal{F}}$. A common choice is the
$L_2$-space with respect to a probability measure $\mu$ on
$\mathcal{X}$. Let $\mathcal{S} \subset \mathcal{F}$ be a subspace
spanned by a finite set of linearly independent functions
$\Psi = \{\psi_1, \dots, \psi_s\}$ (the dictionary).


The goal of data-driven approximation is to find a finite-dimensional
operator $K: \mathcal{S} \to \mathcal{S}$ that best represents the
action of $\mathcal{K}$ when restricted to $\mathcal{S}$. Extended
Dynamic Mode Decomposition (EDMD) \cite{lQL-FD-EMB-IGK:17} provides
the optimal approximation in the $L_2$-sense by orthogonally
projecting the image of the subspace back onto itself.  Formally, let
$P_{\mathcal{S}}: \mathcal{F} \to \mathcal{S}$ denote the orthogonal
projection operator onto $\mathcal{S}$. The EDMD approximation is
given by $K_{\text{EDMD}} \triangleq P_{\mathcal{S}} \mathcal{K}|_{\mathcal{S}}$.

Given a dataset of snapshot pairs $\{(x_i, x_i^+)\}_{i=1}^N$ where
$x_i^+ = T(x_i)$, we construct the data matrices
$\Psi(X), \Psi(X^+) \in \mathbb{R}^{M \times s}$, where the $i$-th rows
are the evaluations of the dictionary functions at $x_i$ and $x_i^+$
respectively.
%
%
The matrix representation of $K_{\text{EDMD}}$ in the
basis $\Psi$ is the solution to the least-squares problem
\begin{align}
  \label{eq:LS-EDMD}
  \min_K \|\Psi(X) K - \Psi(X^+)\|_F,   
\end{align}
given explicitly by
$K = \Psi(X)^\dagger \Psi(X^+) \in \mathbb{R}^{s \times s}$.

This projection interpretation of EDMD reveals a critical limitation:
if the subspace $\mathcal{S}$ is not invariant under $\mathcal{K}$
(i.e., $\mathcal{K}\mathcal{S} \not\subseteq \mathcal{S}$), the
projection $P_{\mathcal{S}}$ discards the component of the dynamics
that evolves orthogonal to $\mathcal{S}$, leading to approximation
errors.

\subsection{Principal Angles and Vectors}\label{sec:principal-angles}
Here we recall the basic definitions and properties of principal
angles and vectors~\cite{AB-GHG:73},
%
%
which will be useful to rigorously quantify the alignment between the
chosen subspace and its evolution under the Koopman operator.

\begin{definition}[Principal Angles and Vectors]
\label{defn:pa_pv}
Let \((\mathcal{H},\langle\cdot,\cdot\rangle)\) be a Hilbert space,
and let \(\mathcal{U},\mathcal{V}\subset\mathcal{H}\) be subspaces
with \(\dim(\mathcal{U})=d_1\) and \(\dim(\mathcal{V})=d_2\).  The
principal angles
$0 \leq \theta_1 \leq \cdots \leq \theta_k \leq \frac{\pi}{2}$
between $\mathcal{U}$ and $\mathcal{V}$, where
$k = \min\{d_1, d_2\}$, are defined recursively as follows:
\begin{align*}
  \cos \theta_j = \max_{u \in \mathcal{U},\, v \in \mathcal{V}} \;
  &
    |\langle
    u,
    v
    \rangle|
  \\ 
  \text{subject to} \quad & \|u\| = \|v\| = 1,
  \\
  & \langle u, u_i \rangle = 0, \; \langle v, v_i \rangle = 0, \,\,
    i = 1, \ldots, j-1, 
\end{align*}
where $u_i, v_i$ are the principal vectors corresponding to the
previous $(j-1)$ angles. The vectors $(u_j, v_j)$ achieving the
maximum are called the $j$-th pair of principal vectors.
\end{definition}

\begin{remark}
  The principal vectors $\{u_j\}_{j=1}^k$ and $\{v_j\}_{j=1}^k$ are
  orthonormal, i.e.,
  \begin{equation*}  
    \langle u_i,u_j\rangle=\delta_{ij},\quad
    \langle v_i,v_j\rangle=\delta_{ij},\quad i,j=1,\dots,k.
  \end{equation*}
  In particular, the principal vectors can be extended to bases of
  their subspaces. For instance, if $k=\dim(\Uc)$, then
  $\{u_j\}_{j=1}^k$ is an orthonormal basis of $\Uc$, and
  $\{v_j\}_{j=1}^k$ can be augmented to an orthonormal basis of
  $\Vc$. \oprocend
\end{remark}

One can show, cf.~\cite[Proposition~4.4]{MH-JC:24-csl-arxiv-revised}, that the
principal vectors $\{u_j\}_{j=1}^k$ and $\{v_j\}_{j=1}^k$ satisfy
$\langle u_i,v_j\rangle=\delta_{ij} \, \cos \theta_i$, for all
$i,j=1,\dots,k$. Next, we describe how to compute the principal angles
and vectors in the Euclidean setting, i.e., $\mathcal{H} = \real^n$,
via the Singular Value Decomposition (SVD).
%
%

\begin{theorem}[Computation via SVD
  \cite{AB-GHG:73}]\label{thm:Golub1973}
  Let $\Uc,\Vc\subset\real^n$ have orthonormal basis matrices
  $Q_{\Uc}\in\real^{n\times d_1}$ and $Q_{\Vc}\in\real^{n\times d_2}$,
  and set $k=\min\{d_1,d_2\}$.  Compute the compact SVD
  %
  %
  \[
    \widetilde U \, \Sigma \, \widetilde V^\top \;=\; Q_{\Uc}^\top Q_{\Vc} ,
    \qquad
    \Sigma=\operatorname{diag}(\sigma_1,\ldots,\sigma_k),
  \]
  %
  %
  where $\sigma_1\ge\cdots\ge\sigma_k\ge0$.  The principal
  angles $\{\theta_j\}_{j=1}^k$ and vectors $\{u_j\}_{j=1}^k$,
  $\{v_j\}_{j=1}^k$ between $\Uc$ and $\Vc$ satisfy
  \begin{equation}
    \cos\theta_j=\sigma_j,\,u_j = Q_{\Uc}\,\widetilde u_j,\, v_j =
    Q_{\Vc}\,\widetilde v_j,  j=1,\dots,k,   
  \end{equation}
  where $\widetilde u_j$ and $\widetilde v_j$ denote the $j$-th
  columns of $\widetilde U \in \real^{d_1\times k}$ and
  $\widetilde V \in \real^{d_2\times k}$. If $\sigma_j=0$, then
  $\theta_j=\tfrac{\pi}{2}$, and the corresponding vectors may be
  chosen from the appropriate nullspaces.
\end{theorem}

\subsection{Invariance Proximity}
As reasoned in Section~\ref{sec:EDMD}, to ensure accuracy of the
Koopman approximation via EDMD, the chosen finite-dimensional subspace
$\mathcal{S}$ should be as close to invariant as possible. We quantify
this ``closeness'' using the geometric concept of principal angles.

Let
$\mathcal{K}\mathcal{S} = \text{span}\{\mathcal{K}\phi \mid \phi \in
\mathcal{S}\}$ denote the image of the subspace under the
operator and let $\{\theta_i\}_{i=1}^s \subset [0, \pi/2]$ be the principal angles between $\mathcal{S}$ and $\mathcal{K}\mathcal{S}$.

\begin{definition}[Invariance
  Proximity~\cite{MH-JC:24-csl-arxiv-revised}]\label{definition:IPT} 
  The invariance proximity of a subspace $\mathcal{S}$ with respect to
  the operator $\mathcal{K}$ is defined by
  $\delta(\mathcal{S}) \triangleq \sin \theta_{\max}(\mathcal{S},
  \mathcal{K}\mathcal{S})$. \oprocend
\end{definition}
%
%

A value of $\delta(\mathcal{S}) = 0$ means that
$\mathcal{K}\mathcal{S} \subseteq \mathcal{S}$, indicating that
$\mathcal{S}$ is an invariant subspace.  The following result
clarifies the practical significance of invariance proximity by
showing that it exactly corresponds to the worst-case prediction error
of the EDMD model over the subspace.

\begin{theorem}[Worst-Case Relative Prediction Error \cite{MH-JC:24-csl-arxiv-revised}]
  Let $\mathcal{S} \subset \mathcal{F}$ be a finite-dimensional
  subspace and let
  $K_{\text{EDMD}} = P_{\mathcal{S}} \mathcal{K}|_{\mathcal{S}}$ be
  the EDMD approximation of the Koopman operator on
  $\mathcal{S}$. Then, 
  \begin{equation}
    \delta(\mathcal{S}) = \sup_{\substack{f \in \mathcal{S}
        \\
        \|\mathcal{K}f\| \neq 0}} \frac{\| \mathcal{K}f -
      K_{\text{EDMD}}f \|}{\| \mathcal{K}f \|}. 
  \end{equation}
\end{theorem}

This result implies that minimizing invariance proximity
$\delta(\mathcal{S})$ directly minimizes the maximum relative error
incurred by the EDMD predictor for any function in the subspace.

\section{Problem Statement}


The overarching goal of this work is to construct effective
finite-dimensional Koopman models for unknown dynamical
systems. Building upon the framework in \cite{MH-JC:25-access}, we
cast this goal as a subspace search problem guided by the metric of
invariance proximity. Let $\mathcal{K}$ denote the Koopman operator
for the dynamics in \eqref{eqn:sys_dynamics}, and let
$\mathcal{S}_{\text{init}}$ represent a broad initial subspace
generated by a finite dictionary of candidate functions. Ideally, one
would extract a strictly invariant subspace
$\mathcal{S} \subseteq \mathcal{S}_{\text{init}}$ capable of accurate
state reconstruction. Because exact invariance is practically
unattainable with finite dictionaries, we relax this requirement to
find a subspace that obtains a user-defined error tolerance. To maintain
the expressivity of the resulting Koopman model, we additionally
require this approximately invariant subspace to have the maximum
possible dimension. We formalize this objective using invariance
proximity as follows.

\begin{problem}[Invariant Subspace Search]\label{problem:subspace_search}
  Given an initial subspace $\mathcal{S}_{\text{init}}$ and a
  tolerance $\epsilon \in [0, 1)$, find a subspace
  $\mathcal{S}^* \subseteq \mathcal{S}_{\text{init}}$ of the largest
  possible dimension such that:
  \begin{equation}\label{eq:problem-bound}
    \delta(\mathcal{S}^*) \le \epsilon.
  \end{equation}
\end{problem}

Finding an exact solution to Problem \ref{problem:subspace_search} is
computationally intractable due to the combinatorial explosion of
searching all possible subspaces. Therefore, our goal shifts to
designing efficient algorithms capable of isolating sufficiently
high-dimensional subspaces that meet the specified invariance
proximity tolerance, even if they fall short of the absolute maximum
dimension. We aim to maximize this dimension primarily to maintain the
model's expressivity. Should the identified subspace $\mathcal{S}^*$
prove inadequate for state reconstruction, one must either expand the
initial candidate dictionary $\mathcal{S}_{\text{init}}$ or accept a
larger error margin by increasing the tolerance $\epsilon$.

To solve Problem~\ref{problem:subspace_search}, the algorithms
introduced below systematically improve invariance proximity by
iteratively pruning non-invariant directions from
$\mathcal{S}_{\text{init}}$. While one could alternatively minimize
$\delta$ by directly learning the dictionary functions $\Psi$ via
neural networks, such optimization landscapes are typically
non-convex, computationally expensive, and devoid of strict
deterministic guarantees. In contrast, our subspace-pruning strategy
leverages the linear algebraic properties of the function space to
surgically eliminate leaky directions from a fixed dictionary. This
approach yields highly efficient algorithms equipped with rigorous
deterministic bounds.

%
%


\section{Subspace Pruning via Principal
  Vectors}\label{sec:pruning_algorithms}

In this section, we propose a systematic approach to solve Problem
\ref{problem:subspace_search} by iteratively refining the
dictionary. Before proceeding, we clarify the standing assumptions and
notation used throughout the remainder of the paper.

\noindent \textbf{Standing Assumption:} Throughout the remainder of
this paper, we consider a finite-dimensional subspace
$\mathcal{S} \subset \mathcal{F}$ and its image
$\mathcal{K}\mathcal{S}$ under the Koopman operator, with
$\dim(\mathcal{S}) = \dim(\mathcal{K}\mathcal{S}) = s$. We denote
their principal angles, arranged in increasing order, and the
corresponding principal vectors by
\begin{equation}
\{\theta_i \}_{i=1}^s \subset [0, \pi/2], \, \{u^{\Sc}_i\}_{i=1}^s \subset \Sc, \, \{\Kc v_i^{\Kc \Sc}\}_{i=1}^s \subset \Kc \Sc.
\label{eqn:setting_principal_arguments}
\end{equation}
These arguments are computed with respect to the inner product $\langle \cdot, \cdot \rangle_{\mathcal{F}}$ on $\mathcal{F}$, for example the $L_2$ inner product induced by the data measure. 

We begin by describing how to compute principal angles and vectors in the
data-driven $L_2(\mu_X)$ setting, which is the first step in our proposed pruning algorithm. 

\subsection{Computation of Principal Angles and Vectors in
  $L_2(\mu_X)$}

Consider the nonlinear system~\eqref{eqn:sys_dynamics}. We utilize a
dataset of $N$ snapshot pairs organized into data matrices
$X, X^+ \in \mathbb{R}^{N \times n}$, where
$X = [x_1, \dots, x_N]^\top$ and $X^+ = [x_1^+, \dots, x_N^+]^\top$
with $x_i^+ = T(x_i)$ for $i=1,\dots,N$. We fix an initial dictionary
of observables $\Psi = [\psi_1, \dots, \psi_s]$ and define the lifted
data matrices $A = \Psi(X) \in \mathbb{R}^{N \times s}$ and
$B = \Psi(X^+) \in \mathbb{R}^{N \times s}$. In this data-driven
framework, we equip the space of observables with the $L_2(\mu_X)$
inner product induced by the empirical data measure $\mu_X = \frac{1}{N} \sum_{i=1}^N \delta_{x_i}$. Specifically, for observables $f,g \in L_2(\mu_X)$, this inner product
is given by
$$
\langle f, g \rangle_{L_2(\mu_X)} = \int_{\mathcal{X}} f(x) g(x) \,
d\mu_X(x) = \frac{1}{N} \sum_{i=1}^N f(x_i) g(x_i).
$$
Under this inner product, we define the discrete evaluation map
$\mathcal{E}: L_2(\mu_X) \rightarrow \mathbb{R}^N$ as:
\begin{equation}
  \mathcal{E}(f) = \begin{bmatrix} f(x_1) \, \cdots \, f(x_N) \end{bmatrix}^\top.
\label{eq:isomorpshism_map}  
\end{equation}
The evaluation map $\Ec|_{\Sc}$ forms an isomorphism between the
candidate function subspace $\mathcal{S}$ and the Euclidean column
space $\mathcal{R}(A)$, and similarly the map $\Ec|_{\Kc \Sc}$ forms
an isomorphism between $\mathcal{K}\mathcal{S}$ and the column space
$\mathcal{R}(B)$. We utilize this isomorphism throughout the paper,
allowing us to compute functional geometric properties---such as
principal angles and orthogonal projections---directly via their
finite-dimensional Euclidean representations.

Specifically, let
$S = \text{span}(\Psi) \subset \Fc \subseteq L_2(\mu_X)$ be the
subspace spanned by the dictionary. Under the map
\eqref{eq:isomorpshism_map},
\begin{equation}
  \Sc \equiv \Rc(A), \quad \Kc \Sc \equiv \Rc(B).
  \label{eq:iso_equiv}
\end{equation}
We can compute the principal angles and vectors between $\Rc(A)$ and
$\Rc(B)$ using the SVD-based procedure from
Theorem~\ref{thm:Golub1973}. As the following result shows, this gives
us the corresponding elements between $\Sc$ and $\Kc \Sc$.

\begin{proposition}\longthmtitle{Computation of principal angles and vectors in $L_2(\mu_X)$}\label{prop:pa_L2}
Let $\{A u_i^A\}_{i=1}^{{s}} \subset \Rc(A)$,
$\{B v_i^B\}_{i=1}^{{s}} \subset \Rc(B)$, and
$\{\theta_i(\Rc(A), \Rc(B))\}_{i=1}^{{s}}$ be the principal vectors
and angles between $\Rc(A)$ and $\Rc(B)$, where
$u_i^A, v_i^B \in \real^s$ are the coefficients of the principal
vectors in the bases of $\Rc(A)$ and $\Rc(B)$, resp. 
Then,
\begin{enumerate}
\item $\theta_i(\Sc, \Kc \Sc) = \theta_i(\Rc(A), \Rc(B))$, for all
  $i \in [s]$,
\item $ u_i^{\Sc}(\cdot) = \Psi(\cdot) u_i^A$,
  $v_i^{\Kc \Sc}(\cdot) = \Psi(\cdot) v_i^B$, for all $ i \in [s]$.
\end{enumerate}
\end{proposition}
\begin{proof}
  Since $\Ec|_{\Sc}$ and $\Ec|_{\Kc \Sc}$ are isomorphisms, the
  principal angles between $\Sc$ and $\Kc \Sc$ are the same as those
  between $\Rc(A)$ and $\Rc(B)$ due to the
  equivalence~\eqref{eq:iso_equiv}. This establishes the first
  claim. For the second claim, it is easy to verify that
  $\Ec(u_i^{\Sc}) = A u_i^A$ and $\Ec(\Kc v_i^{\Kc \Sc}) = B v_i^B$
  and hence, $\{u_i^{\Sc}\}, \{\Kc v_i^{\Kc \Sc}\}$ satisfy
  Definition~\ref{defn:pa_pv} and are the principal vectors between
  $\Sc$ and $\Kc \Sc$.
\end{proof}

The above discussion clarifies how to compute principal arguments for
the specific $L_2(\mu_X)$ inner product used in the data-driven
setting.
Leveraging this finite-dimensional identification, the algorithms
presented below proceed by iteratively pruning the subspace
$\mathcal{S}$ to find a target subspace $\mathcal{S}^*$ that satisfies
the invariance bound~\eqref{eq:problem-bound}.


\subsection{Single-Principal-Vector (SPV) Pruning}
Our solution to Problem \ref{problem:subspace_search} leverages the geometric interpretation of invariance proximity. Recall from Definition~\ref{definition:IPT} that $\delta(\Sc)$ is determined by the largest principal angle between the subspace $\Sc$ and its image $\Kc\Sc$. The principal vectors associated with these large angles pinpoint where the dynamics ``leak'' out of the subspace. To address this, the SPV algorithm operates on a ``worst-offender'' principle: we iteratively remove the specific dimension most responsible for violating invariance. At each step, we identify the principal vector $u_{\max}^{\Sc} \in \Sc$ corresponding to the maximum principal angle $\theta_{\max}$. By projecting the current subspace onto the orthogonal complement of $u_{\max}^{\Sc}$, we surgically remove this direction. This process is repeated until the maximum angle falls below a desired tolerance $\epsilon$, as formalized in Algorithm \ref{alg:naive-pruning}.

\begin{algorithm}[h]
  \caption{\textbf{SPV Pruning}}
  \label{alg:naive-pruning}
  \begin{algorithmic}[1]
    \REQUIRE $\Sc, \Kc \Sc \subset \Fc$, $\epsilon \in [0,1]$
    \STATE Initialize $\Sc_1 \gets \Sc$, \, $i \gets 0$
    \WHILE{True}
    \STATE $i \gets i + 1$
    \IF{$\Sc_{i} = \emptyset$} 
    \RETURN $\emptyset$ \RightComment{Terminate with failure}
    \ENDIF
    \STATE $\{u^{\Sc}_j\}, \{\theta_j\} \gets \text{Principal
      arguments}(\Sc_i, \Kc \Sc_i)$ 
    \smallskip
    \IF{$\sin \theta_{\max} \leq \epsilon$}
    \RETURN $\Sc_i$ \RightComment{Terminate with success}
    \ENDIF
    \STATE $\Sc_{i+1} \gets \Sc_i \setminus \text{span}(u_{\max}^{\Sc})$
    \ENDWHILE
  \end{algorithmic}
\end{algorithm}

\subsection{Equivalence of RFB-EDMD and SPV
  Algorithms}\label{sec:equivalence}

Here, we explain the equivalence of SPV pruning with the Recursive
Forward-Backward EDMD (RFB-EDMD) algorithm introduced
in~\cite{MH-JC:25-access}.  To do so, we start by introducing key
ingredients of RFB-EDMD. Given data matrices
$X, X^+ \in \mathbb{R}^{N \times n}$ from the nonlinear
system~\eqref{eqn:sys_dynamics}, consider the standard ``forward''
EDMD matrix $K_f = \Psi(X)^\dagger \Psi(X^+)$ (corresponding to the
forward-in-time evolution $x \to x^+$) and the ``backward'' EDMD
matrix $K_b = \Psi(X^+)^\dagger \Psi(X)$ (corresponding to the
backward-in-time evolution $x^+ \to x$).  Let $M_c = I - K_f K_b$ be
the \textit{consistency matrix} measuring the discrepancy between the
forward and backward predictions.
The next result establishes that the eigenvalues of the consistency
matrix are exactly the squared sines of the principal angles between
the search space $\mathcal{S}$ and its image
$\Kc \mathcal{S}$.
%
%
Furthermore, the result identifies the eigenvectors of
$M_c$ as the coefficients of the principal vectors in~$\Sc$.

\begin{lemma}\longthmtitle{Spectral Characterization of
    Consistency}\label{lem:spectral-consistency}
  Let $\Sc \subset \Fc$ be the subspace spanned by the
  dictionary~$\Psi = [\psi_1, \dots, \psi_s]$. Let
  $A = \Psi(X) \in \real^{N \times s}$ and
  $B = \Psi(X^+) \in \real^{N \times s}$ be the data matrices
  representing the domain and image of the Koopman operator on
  $\mathcal{S}$, with full column rank. Then, the consistency matrix
  $M_c = I - K_f K_b$ satisfies:
  \begin{enumerate}
  \item Its eigenvalues $\{\lambda_i \}_{i=1}^s$ are squared sines of
    the principal angles, i.e., $ \lambda_i = \sin^2 \theta_i$ 
    $\forall \,\, i \in [s]$;

  \item Its eigenvectors $\{v_i\}_{i=1}^s$ correspond to the principal
    vectors of $\mathcal{S}$, as specified by
    $u^{\Sc}_i(\cdot) = \Psi(\cdot)v_i$, $\forall \,\,i \in [s]$.
  \end{enumerate}
\end{lemma}
\begin{proof}
  Using the definition of the forward $ K_f = A^\dagger B$ and
  backward $ K_b = B^\dagger A $ EDMD matrices, we have that
  $M_c = I - K_f K_b = I - A^\dagger B B^\dagger A$. Note that
  $P_B = B B^\dagger$ is the orthogonal projection onto
  $\mathcal{R}(B)$. Therefore, $M_c = I - A^\dagger P_B A$.  Consider
  the QR decompositions $A = Q_A R_A$ and $B = Q_B R_B$, where
  $Q_A, Q_B$ have orthonormal columns and $R_A, R_B$ are invertible
  upper triangular matrices. Substituting into the expression for
  $M_c$, we get
  \begin{align*}
    M_c &= I - (R_A^{-1} Q_A^{\top}) (Q_B Q_B^{\top}) (Q_A R_A) \\
        &= I - R_A^{-1} (Q_A^{\top} Q_B Q_B^{\top} Q_A) R_A .
  \end{align*}
  The term $Q_A^{\top} Q_B$ is the matrix of inner products between the
  orthonormal bases. Following Theorem~\ref{thm:Golub1973}, let the SVD of
  this matrix be $ Q_A^{\top} Q_B = U_A (\cos \Theta) V_B^{\top} $, where
  $\cos \Theta = \text{diag}(\cos \theta_1, \dots, \cos \theta_s)$
  contains the cosines of the principal angles. Substituting this SVD
  back, we obtain:
  \begin{align*}
    Q_A^{\top} Q_B Q_B^{\top} Q_A
    = U_A
      \cos^2 \Theta U_A^{\top} .
  \end{align*}
  Therefore, the consistency matrix becomes:
  \begin{align*}
    M_c &= I - R_A^{-1} (U_A \cos^2 \Theta U_A^{\top}) R_A
    \\
        &= R_A^{-1} (I - U_A \cos^2 \Theta U_A^{\top}) R_A \\
        &= R_A^{-1} U_A (I - \cos^2 \Theta) U_A^{\top} R_A .
  \end{align*}
  Since $I - \cos^2 \Theta = \sin^2 \Theta$, we have
  \begin{align*}
    M_c = (R_A^{-1} U_A) \sin^2 \Theta (R_A^{-1} U_A)^{-1}.    
  \end{align*}
  This similarity relation proves that the eigenvalues of $M_c$ are
  exactly $\{\sin^2 \theta_i\}_{i=1}^s$. Furthermore, the eigenvectors
  of $M_c$ are the columns of $R_A^{-1} U_A$. The $i$-th eigenvector
  $v_i = R_A^{-1} u^A_i$ of $M_c$ satisfies
  \begin{align*}
    \Psi(X) v_i = A v_i = Q_A R_A (R_A^{-1} u^A_i) = Q_A u^A_i.    
  \end{align*}
  Since $u^A_i$ is the left singular vector of $Q_A^{\top} Q_B$, according
  to Theorem~\ref{thm:Golub1973}, the vector $Q_A u^A_i = A v_i$ is
  precisely the $i$-th principal vector of the subspace
  $\mathcal{R}(A)$. Utilizing Proposition~\ref{prop:pa_L2}, we have
  $u^{\Sc}_i(\cdot) = \Psi(\cdot)v_i \,\, \forall \, i \in [s]$.
\end{proof}

We leverage this spectral characterization to establish the algebraic
equivalence between the RFB-EDMD and SPV pruning algorithms.

\begin{theorem}[Algorithmic Equivalence]\label{thm:RFB-EDMD_SPV}
  The SPV pruning and RFB-EDMD algorithms are algebraically equivalent, removing the identical one-dimensional subspace from the search space $\mathcal{S}$ at every iteration.
\end{theorem}

\begin{proof}
  At iteration $k$, RFB-EDMD updates the search space via $\mathcal{S}_{k+1} = \mathcal{S}_k \setminus \text{span}\{\Psi(\cdot)v_{\max}\}$, where $v_{\max}$ is the eigenvector of the consistency matrix $M_c$ corresponding to its largest eigenvalue $\lambda_{\max}$. Similarly, SPV updates via $\mathcal{S}_{k+1} = \mathcal{S}_k \setminus \text{span}\{u_{\max}^{\Sc}\}$, where $u_{\max}^{\Sc} \in \mathcal{S}_k$ is the principal vector corresponding to the largest principal angle $\theta_{\max}$ between $\mathcal{S}_k$ and its image $\Kc \Sc_k$.

  Lemma~\ref{lem:spectral-consistency} establishes two key equivalences: $\lambda_{\max} = \sin^2 \theta_{\max}$ and $\Psi(\cdot)v_{\max} \equiv u_{\max}^{\Sc}$. Because $\sin^2 \theta$ is monotonically increasing on $[0, \pi/2]$, maximizing the eigenvalue directly maximizes the principal angle. Consequently, both algorithms discard the exact same functional direction at every step, producing an identical sequence of nested subspaces $\mathcal{S}_0 \supset \mathcal{S}_1 \supset \dots$ and the same final output.
\end{proof}

To conclude, the SPV and RFB-EDMD algorithms are indeed equivalent when working with the $L_2(\mu_X)$ inner product. However, the SPV algorithm is more general and can be applied to arbitrary inner products, as it relies solely on the geometric notion of principal angles and vectors. This generality of SPV allows us to extend the pruning framework to a wider class of function spaces and inner products, which we will explore in future work.

\section{Efficient Computation of Principal Angles and Vectors via
  Rank-One Updates}\label{sec:efficient_computation}


The naive SPV pruning algorithm is computationally expensive, especially for large dictionaries, because it requires a full Singular Value Decomposition (SVD) at each iteration to recompute principal angles and vectors after a direction is dropped. To address this bottleneck, this section develops an efficient update procedure based on symmetric rank-one corrections to an eigenproblem. While our theoretical results apply broadly to any general inner product space $\Fc$, our algorithmic implementation and time complexity analyses focus specifically on the $L_2(\mu_X)$ setting, which is the most common in data-driven Koopman analysis.

\subsection{Computation of Principal
  Arguments}\label{sec:computation_principal_angles}

Consider the subspaces $\Sc, \Kc \Sc \subset \Fc$ according to \eqref{eqn:setting_principal_arguments}. Let
\begin{align*}
  \Uc & = [u^{\Sc}_1 \, u^{\Sc}_2 \, \dots \, u^{\Sc}_s], \quad
        \Kc \Uc = [\Kc u^{\Sc}_1 \, \Kc u^{\Sc}_2 \, \dots \, \Kc
        u^{\Sc}_s],
  \\
  \Lambda_{\cos}
      & \!=\! \text{diag}(\cos \theta_1, \dots, \cos \theta_s)
        , \,
  \Lambda_{\sin} \!=\! \text{diag}(\sin \theta_1, \dots, \sin
  \theta_s). 
\end{align*}
Define the pruned subspace
$\Sc^{\text{new}} = \text{span}(u^{\Sc}_1, u^{\Sc}_2, \dots,
u^{\Sc}_{s-1})$ obtained by dropping the top $k$ principal vectors.
The updated image is
$\Kc \Sc^{\text{new}} = \text{span}(\Kc u^{\Sc}_1, \Kc u^{\Sc}_2,
\dots, \Kc u^{\Sc}_{s-1})$.

We introduce a vector $\omega \subset \Kc \Sc$ that spans the orthogonal
complement of the new image subspace $\Kc \Sc^{\text{new}}$ within
$\Kc \Sc$. This helps us decompose the projection onto the
new image subspace $\Kc \Sc^{\text{new}}$ into a projection onto the
old image subspace $\Kc \Sc$ followed by a rank-one
update, which can be computed efficiently.
%
%
In order to improve numerical stability, we do not compute this vector via explicit projections, which can be numerically unstable. Instead, this is extracted from the thin QR decomposition
\cite[Section 5.2]{GHG-CFVL:13}
%
%
of the image matrix $ \Kc \Uc$,
\begin{equation}
  W R = [\Kc u^{\Sc}_1, \Kc u^{\Sc}_2, \dots, \Kc u^{\Sc}_s],
  \label{eq:QR_W}
\end{equation}
where $R \in \real^{s \times s}$ is an upper triangular matrix and
$W = [w_1 \, w_2 \, \dots \, w_s]$ is a matrix with orthonormal
columns.
%
%
Because the first $s-1$ columns of $\Kc \Uc$ span
$\Kc \Sc^{\text{new}}$, the standard QR algorithm guarantees that the
first $s-1$ columns of $W$ form an orthonormal basis for
$\Kc \Sc^{\text{new}}$. The remaining column spans the orthogonal
complement of $\Kc \Sc^{\text{new}}$ within~$\Kc \Sc$. Thus, we define
$\omega = w_s$.

Let
$\omega = \sum_{i=1}^{s} d_i^{\omega} \Kc v^{\Kc \Sc}_i$ be a basis expansion, and
define
$d^{\omega} = [d^{\omega}_1 \, \dots \,
d^{\omega}_s]^{\top}$. By construction, the coordinate vectors are
orthonormal and satisfy
$(d^{\omega_l})^{\top} d^{\omega_l} = \delta_{ll}$.  Define the rank-one update matrices
$N_0, N_1 \in \mathbb{R}^{s \times s}$ by
\begin{align}
  N_0 \!=\! \Lambda_{\sin}^2, \,
  N_1 \!=\! N_{0} + \Lambda_{\cos} d^{\omega} (\Lambda_{\cos}
  d^{\omega})^{\top}. 
    \label{eq:rank_one_matrices}
\end{align}

The following result describes to how to compute the principal angles
and vectors between $\Sc^{\text{new}}$ and~$\Kc \Sc^{\text{new}}$
efficiently using the eigenpairs of~$N_1$.

\begin{theorem}\longthmtitle{Efficient Computation of Principal
  Arguments}\label{thm:efficient_computation}
  Let $\tilde{N}_1 \in \mathbb{R}^{(s-1) \times (s-1)}$ be the
  truncated matrix obtained by dropping the last row and column
  of $N_1$. Let
  $(\lambda_{\alpha} \in \mathbb{R}, z_{\alpha} \in
  \mathbb{R}^{s-1})_{\alpha = 1}^{s-1}$ be the eigenpairs of
  $\tilde{N}_1$, arranged with increasing eigenvalues.  Then, the
  principal vectors
  $\{u^{\Sc^{\text{new}}}_{\alpha}\}_{\alpha=1}^{s-1} \subset
  \Sc^{\text{new}}$ and squared principal sines
  $\sin^2 \theta_{\alpha}(\Sc^{\text{new}}, \Kc \Sc^{\text{new}})$ are
  given by
  \begin{equation}
    \begin{aligned}
      \sin^2 \theta_{\alpha} (\Sc^{\text{new}}, \Kc \Sc^{\text{new}})
      =
        \lambda_{\alpha},
      \,\,
      u^{\Sc^{\text{new}}}_{\alpha} \!= \Uc_{s-1} z_{\alpha}, \label{eq:new_vectors_angles} 
    \end{aligned}
  \end{equation}
  for $\alpha \in [s-1]$, where
  $\Uc_{s-1} = [u^{\Sc}_1 \, u^{\Sc}_2 \, \dots \,
  u^{\Sc}_{s-1}]$.
\end{theorem}
\begin{proof}
  Based on the construction of $\omega$ we have the
  orthogonal direct sum
  $\Kc \Sc = \Kc \Sc^{\text{new}} \,\oplus\, \text{span}(\omega
  )$. Thus, the orthogonal projection operator can be
  expressed as
  \[
    \mathcal{P}_{\Kc \Sc^{\text{new}}} = \mathcal{P}_{\Kc \Sc} -
    \mathcal{P}_{\omega}.
  \]
  Let $\tilde{u} \in \Sc^{\text{new}}$ be decomposed as
  $\tilde{u} = \sum_{\alpha=1}^{s-1} c_\alpha u^{\Sc}_{\alpha}$.
  %
  %
  The projection of $\tilde{u}$ onto the original image
  space $\Kc \Sc$ yields
  $
    \mathcal{P}_{\Kc \Sc} \tilde{u} = \sum_{\alpha=1}^{s-1} c_\alpha \Kc v^{\Kc
      \Sc}_{\alpha} \cos \theta_{\alpha}.
  $
  To compute the projection onto $\omega$, we let
  $t = \langle \omega, \tilde{u} \rangle_{\Fc}$. Then,
  \begin{align*}
    t &= \Big\langle \sum_{i=1}^{s} d^{\omega_l}_i \Kc v^{\Kc \Sc}_i,
    \sum_{\alpha=1}^{s-1} c_\alpha u^{\Sc}_{\alpha} \Big\rangle_{\Fc} \\
    &= \sum_{\alpha=1}^{s-1} d^{\omega_l}_\alpha c_\alpha \cos
      \theta_{\alpha} = \tilde{c}^{\top} 
    \Lambda_{\cos} d^{\omega_l}, 
  \end{align*}
  where $\tilde{c} = [c_1, \dots, c_{s-1}, 0]^{\top} \in \mathbb{R}^s$.
  The projection is therefore
  $\mathcal{P}_{\Kc \Sc^{\text{new}}} \tilde{u} = \mathcal{P}_{\Kc
    \Sc} \tilde{u} - t \omega$.
  %
  Since the basis elements are orthonormal, the squared norm
  evaluates as
  \begin{align*}
    \| \mathcal{P}_{\Kc \Sc^{\text{new}}} \tilde{u} \|_{\Fc}^2 
    &=   
      \big\| \Lambda_{\cos} \tilde{c} - t d^{\omega} \big\|_2^2
    \\
    &= \tilde{c}^{\top} \Lambda_{\cos}^2 \tilde{c} + t^2 - 2
      t \tilde{c}^{\top} \Lambda_{\cos} d^{\omega}
    \\
    &= \tilde{c}^{\top} \Big( \Lambda_{\cos}^2 - \Lambda_{\cos}
      d^{\omega} (\Lambda_{\cos} d^{\omega})^{\top}  \Big) \tilde{c}.  
  \end{align*}
  Using equation \eqref{eq:rank_one_matrices}, this expression can be
  rewritten as
  $\| \mathcal{P}_{\Kc \Sc^{\text{new}}} \tilde{u} \|_{\Fc}^2 =
  \tilde{c}^{\top} (I - \tilde{N}_1) \tilde{c}$. Since
  $\| \tilde{u} \|_{\Fc}^2 = \tilde{c}^{\top} \tilde{c}$, utilizing
  Lemma~\ref{lemma:cos_theta_max}, we find the principal angles by
  solving the Rayleigh quotient minimization:
  \[
    \min_{\tilde{c} \in \mathbb{R}^{s-k}}
    \frac{\tilde{c}^{\top} (I - \tilde{N}_1) \tilde{c}}{\tilde{c}^{\top}
      \tilde{c}}.
  \]
  %
  %

  Note that Lemma~\ref{lemma:cos_theta_max} dictates that subsequent
  minimizers must be orthogonal in $\Fc$ (i.e.,
  $\langle \tilde{u}_a, \tilde{u}_b \rangle_{\Fc} =
  0$). Orthonormality of the basis $\Uc^{\text{new}}$ guarantees this
  is equivalent to the Euclidean constraint
  $\tilde{c}_a^\top \tilde{c}_b = 0$. By the Courant-Fischer min-max
  theorem, sequentially minimizing a Rayleigh quotient subject to
  Euclidean orthogonality constraints is exactly equivalent to
  computing the eigendecomposition of the symmetric matrix
  $(I-\tilde{N}_{1})$.
  
  The successive minimums of this quotient yield the eigenvalues
  $\cos^2 \theta_\alpha$. Consequently, the eigenvalues of
  $\tilde{N}_1$ are exactly
  $1 - \cos^2 \theta_\alpha = \sin^2 \theta_\alpha$, verifying the angles in
  equation \eqref{eq:new_vectors_angles}. Furthermore, the corresponding
  eigenvectors $z_\alpha$ directly provide the parameterization for
  the updated principal vectors
  $u_{\alpha}^{\Sc^{\text{new}}} = \Uc^{\text{new}}z_\alpha$, yielding
  the vectors in equation \eqref{eq:new_vectors_angles}.
\end{proof}

\begin{remark}\longthmtitle{Applicability to General Inner Product
  Spaces}\label{rem:general_spaces}
  It is important to emphasize that this rank-one update procedure is
  entirely coordinate free and holds for general abstract inner
  product spaces, 
  not just Euclidean spaces. Because the formulation relies purely on
  the principal angles and the coefficients of the QR decomposition,
  the algorithm never requires evaluating the abstract inner product
  $\langle \cdot, \cdot \rangle_{\Fc}$ explicitly during the update
  steps. \oprocend
\end{remark}

\begin{remark}[Computational Complexity and LAPACK]\label{rem:lapack}
  The core of our efficient update procedure relies on the symmetric rank-one modification problem: computing the eigendecomposition of $A \pm \rho u u^{\top}$ given the known eigendecomposition of a symmetric matrix $A \in \mathbb{R}^{n \times n}$. This update requires only $O(n^2)$ operations and is robustly implemented in the LAPACK subroutine \texttt{DLAED9}~\cite{EA-LAPACK:99}. In the context of Theorem~\ref{thm:efficient_computation}, the transition from the original matrix $N_0$ to the updated matrix $N_1$ takes exactly this rank-one form. Consequently, computing the new principal angles and vectors after dropping a single direction costs $O(s^2)$ operations. This provides a significant computational advantage over the $O(s^3)$ cost required to naively recompute the full decomposition at each step. \oprocend
\end{remark}
\subsection{Incremental Basis Update via QR Decomposition}

While rank-one updates efficiently yield the new principal angles,
iterative pruning also requires maintaining an orthogonal basis for
the image space $\Kc \Sc$. If we were to use the empirical $L_2(\mu)$
inner product over $N$ data points, a naive re-computation of the thin
%
%
QR decomposition for the image matrix $\Kc \Uc$ as in~\eqref{eq:QR_W}
would incur a prohibitive computational cost of $O(Ns^2)$ at each
step.

To circumvent this, we update the factors incrementally. Suppose we
have the decomposition $\Kc \Uc = W R$ available. By restricting
operations to the smaller triangular matrix $R$, we avoid processing
the full, high-dimensional matrix $\Kc \Uc$ directly. We propose the
following efficient procedure to compute the QR decomposition of the
updated image $\Kc \Uc^{\text{new}}$.  Recall from
Theorem~\ref{thm:efficient_computation} that the new principal vectors
$\Uc^{\text{new}}$ are formed by taking linear combinations of the
retained basis using the computed eigenvectors, i.e.,
\begin{equation}
  \Uc^{\text{new}} = [u^{\Sc^{\text{new}}}_1 \, \dots \,
  u^{\Sc^{\text{new}}}_{s-1}] = \underbrace{[u^{\Sc}_1 \, \dots \,
    u^{\Sc}_{s-1}]}_{\text{Retained Basis}} \underbrace{[z_1 \, \dots
    \, z_{s-1}]}_{E}. 
\end{equation}

\textbf{1. Construct the Re-alignment Matrix}\\
Let $E \in \mathbb{R}^{(s-1) \times (s-1)}$ be the matrix of
eigenvectors. We construct a transformation matrix
$T \in \mathbb{R}^{s \times (s-1)}$ by padding $E$ with zeros to align
with the original $s$-dimensional space:
\begin{equation}
  T = \begin{bmatrix} E \\ 0_{1 \times (s-1)} \end{bmatrix}.
\end{equation}

\textbf{2. Update the Triangular Factor}\\
We apply the transformation $T$ to the existing upper triangular
factor $R$ to form the intermediate matrix
$C = R T \in \mathbb{R}^{s \times (s-1)}$. We then perform a QR
decomposition of $C$ as
\begin{equation}
  C = Q_C R_C ,
\end{equation}
where $Q_C \in \mathbb{R}^{s \times (s-1)}$ is orthogonal and
$R^{\text{new}} = R_C \in \mathbb{R}^{(s-1) \times (s-1)}$ is the new
upper triangular factor.

\textbf{3. Update the Orthogonal Bases}\\
Finally, we update the orthogonal image basis $W$ by applying the
rotations derived above,
\begin{align}
    W^{\text{new}} &= W Q_C.
\end{align}
As we show next, the resulting matrices form the QR decomposition of
the new image space $\Kc \Uc^{\text{new}}$

\begin{lemma}[Correctness of Incremental
  QR]\label{lemma:incremental_validity}
  %
  %
  Consider the notation and construction of
  Section~\ref{sec:computation_principal_angles}.  The matrices
  $W^{\text{new}}$ and $R^{\text{new}}$ are a valid QR decomposition
  of~$\Kc \Uc^{\text{new}}$, i.e.,
  $\Kc \Uc^{\text{new}} = W^{\text{new}} R_C$.
\end{lemma}
\begin{proof}
  By equation \eqref{eq:new_vectors_angles}, we have
  $\Uc^{\text{new}} = \Uc T$.
  %
  %
  Linearity of the operator $\Kc$ implies
  $\Kc \Uc^{\text{new}} = \Kc \Uc T$. Substituting the initial QR
  decomposition $\Kc \Uc = W R$
  %
  %
  yields
  \[
    \Kc \Uc^{\text{new}} = W R T.
  \]
  Using the definition of $C$ and its decomposition $C = Q_C R_C$, we
  expand the expression as
  \[
    \Kc \Uc^{\text{new}} = W (Q_C R_C) = (W Q_C) R_C = W^{\text{new}} R^{\text{new}}.
  \]
  Since $W$ has orthonormal columns and $Q_C$ is orthogonal, their
  product $W^{\text{new}}$ also has orthonormal columns. Furthermore,
  $R^{\text{new}} = R_C$ is upper triangular by construction. Thus,
  $W^{\text{new}} R^{\text{new}}$ is a valid QR decomposition.
\end{proof}

Based on Lemma~\ref{lemma:incremental_validity}, instead of computing
the QR decomposition of $\Kc \Uc^{\text{new}}$ from scratch (which
costs $O(N(s-1)^2)$ for $N$ data points), we can compute it using the
existing QR decomposition of $\Kc \Uc$ and the QR decomposition of the
smaller matrix $C$, which costs only $O(s(s-1)^2)$. This incremental
update leads to massive computational savings in data-driven
applications where the number of data points vastly exceeds the
dictionary size ($N \gg s$).

%
%
\begin{algorithm}[h]
  \caption{\textbf{Efficient Computation of Principal Arguments}}
  \label{alg:efficient_computation}
  \begin{algorithmic}[1]
    \REQUIRE 
    $\Uc = [u_1, \dots, u_s]$ \RightComment{\textit{PV}s of $\Sc$}
    \REQUIRE 
    $(W ,R) $ \RightComment{Thin QR of $\Kc \Uc$} 
    \REQUIRE 
    $\Lambda_{\sin \theta} \in \real^{s \times s}$
    \RightComment{Principal sines of $\Sc, \Kc \Sc$}  
    \smallskip
    \STATE $\omega \gets w_s$
    \RightComment{Extract last  column of $W$} 
    \smallskip
    \STATE $\Lambda_{\cos} d^{\omega} \gets \langle \Uc, \omega \rangle_{\Fc}$ \RightComment{Rank-one update vector}
    \smallskip
    \STATE $\Lambda_0 \gets (\Lambda_{\sin \theta}^2)_{1:s-1, \, 1:s-1}$ \RightComment{Top-left block}
    \smallskip
    \STATE $b_1 \gets \text{first } s-1 \text{ elements of } \Lambda_{\cos} d^{\omega}$ 
    \smallskip
    \STATE $(\Lambda_1, E_1) \gets \texttt{DLAED9}(\Lambda_{0}, b_1)$ \RightComment{LAPACK routine} 
    \smallskip    
    \STATE Set $T \!\gets\!
    \begin{bmatrix}
      E_1 \\ 0
    \end{bmatrix},\, C \!\gets\! R T$, $\Lambda_{\sin \theta}^{\text{new}} \gets \Lambda_1^{1/2}, \, \Uc^{\text{new}} \gets \Uc T$
    \smallskip
    \STATE $(Q_C, R_C) \gets \text{QR}(C)$, $W^{\text{new}} \gets W Q_C, \quad R^{\text{new}} \gets R_C$
    \smallskip
    \RETURN $(\Lambda_{\sin \theta}^{\text{new}}, \Uc^{\text{new}}, W^{\text{new}}, R^{\text{new}})$
  \end{algorithmic}
\end{algorithm}


\begin{remark}\longthmtitle{Efficient Algorithm for recomputing
    Principal Arguments}
  Algorithm~\ref{alg:efficient_computation} exploits the results of
  Theorem~\ref{thm:efficient_computation} and
  Lemma~\ref{lemma:incremental_validity} to compute efficiently new
  principal angles and vectors after dropping the top principal vector.  This procedure can be used as a high-speed
  subroutine in SPV pruning.  The algorithm takes as
  input the principal vectors $\Uc$ of the subspace $\Sc$, the
  principal sines $\Lambda_{\sin \theta}$ between $\Sc$ and its image
  $\Kc \Sc$, and the QR decomposition $(W,R)$ of $\Kc \Uc$. The algorithm returns the
  corresponding quantities for the updated subspace $\Sc^{\text{new}}$
  after dropping the top principal vector.

  In step 2, we compute $\Lambda_{\cos} d^{\omega}$
  directly via the inner product $\langle \Uc, \omega \rangle_{\Fc}$
  (which reduces to the matrix multiplication $\Uc^T \omega$ when
  utilizing the empirical $L_2$ inner product). This avoids the need
  to explicitly compute the principal vectors of~$\Kc \Sc$.  In steps
  3 and 4, we restrict the matrices and vectors to their first $s-1$
  dimensions. This truncation corresponds to finding the eigenpairs of
  $\tilde{N}_1$, which is obtained by dropping the last row and
  column of the full update matrix $N_1$. In step 5, we use the
  LAPACK subroutine \texttt{DLAED9} \cite{EA-LAPACK:99} to compute the
  eigenpairs $(\Lambda_1, E_1)$ of the matrix
  $\Lambda_{0}  + b_1b_1^T$ obtained via the
  symmetric rank-one update.
  Section~\ref{sec:sim-computational_efficiency_numerics} describes
  numerical benchmarks demonstrating the vast computational efficiency
  Algorithm~\ref{alg:efficient_computation} when integrated into the
  pruning procedures.  \oprocend
\end{remark}

\section{Simulation Results}\label{sec:sim-computational_efficiency_numerics}

In this section, we compare the computation times of the proposed rank-one update scheme against the naive approach. Additionally, we present numerical results that demonstrate the improved quality of the leading non-trivial Koopman eigenfunction obtained via the SPV pruning procedure. All simulations are performed in Python 3.11.4 on a machine with an Apple M1 Pro chip and 16 GB of RAM. Unless otherwise specified, we employ the standard inner product on $L_2(\mu_X)$, where $X$ denotes the trajectory data.

To benchmark the numerical performance of the proposed pruning algorithms, we consider the damped Duffing oscillator. Using a time step of $\Delta_t = 0.01$, the discretized system dynamics are given by:
\begin{subequations}\label{eq:sys1}
  \begin{align}
    x_{1}^+ &= x_1 + \Delta_t x_2, \\
    x_2^+ &= x_2 + \Delta_t (-0.5 x_2 + x_1 - x_1^3).
  \end{align}
\end{subequations}
This system possesses two stable equilibria at $x = (\pm 1, 0)$ and an unstable equilibrium at the origin. Its rich nonlinear dynamics make it an ideal testbed for evaluating Koopman operator approximation methods.

To assess the efficiency gains of the rank-one update scheme (Algorithm \ref{alg:efficient_computation}), we evaluate its computation time against the naive approach, which recomputes the principal vectors and angles from scratch at every pruning step. Trajectory data is gathered by simulating the system from 500 random initial conditions---sampled uniformly from the box $[-2, 2]^2$---for 100 time steps each. We compare the computation times across dictionary sizes of $s \in \{28, 103, 403\}$, which are constructed using polynomial and radial basis functions. The timing results are summarized in Table \ref{tab:timing_comparison}.
\begin{table}[ht]
  \centering
  \setlength{\tabcolsep}{6pt} 
  
  \begin{tabular}{rcc}
    \toprule
    \textbf{Initial Dim.} & 
    \textbf{SPV} & 
    \textbf{SPV (Rank-1)} \\
    \midrule
    28  & 1.0162   & 0.1725  \\
    103 & 22.2471  & 1.5325  \\
    403 & 134.6234 & 6.4018 \\
    \bottomrule
  \end{tabular}
  \caption{\normalfont Wall-clock time comparison (in seconds) demonstrating the efficiency of rank-one updates across varying dictionary sizes.}
  \label{tab:timing_comparison}
\end{table}

Next, we utilize the same dataset with $N = 50,000$ to evaluate the quality of the leading non-trivial Koopman eigenfunction approximation obtained by SPV pruning. This is the eigenfunction associated with the eigenvalue closest to $\lambda = 1$, excluding the trivial eigenvalue at $1$. We employ thin-plate spline radial basis functions to construct the initial dictionary, with $k_{\text{centers}} = 500$ chosen via $k$-means clustering. We also include polynomial features up to degree $1$, resulting in a total initial dictionary size of $s = 503$. We then apply SPV pruning to reduce the dictionary size from $s = 503$ to $s^* = 15$. The results are visualized in Figure \ref{fig:eigenfuncs}. The improved eigenfunction provides a clear separation of the basins of attraction of the two stable equilibria.

\begin{figure}[htb!]
  \centering
  \begin{subfigure}[b]{0.49\linewidth}
    \centering
    \includegraphics[width=\linewidth]{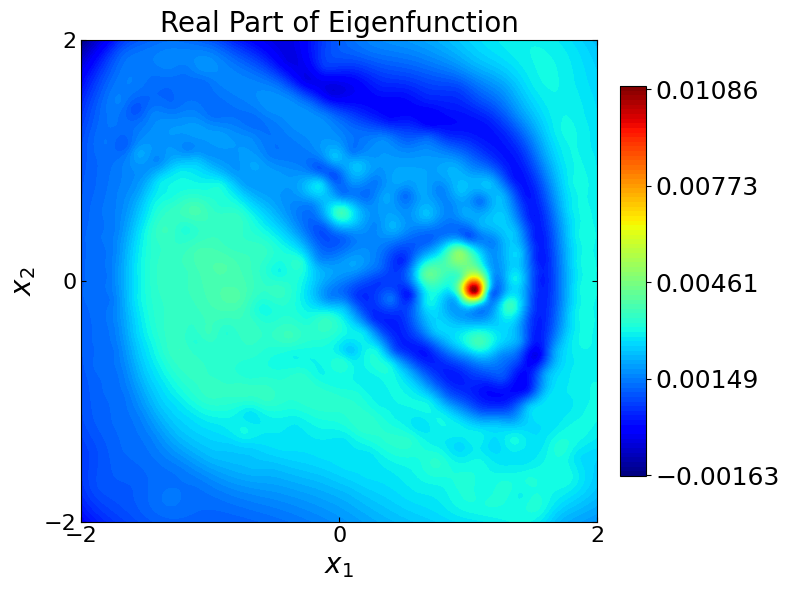}
    \label{fig:unpruned_eigenfunc}
  \end{subfigure}
  \hfill
  \begin{subfigure}[b]{0.49\linewidth}
    \centering
    \includegraphics[width=\linewidth]{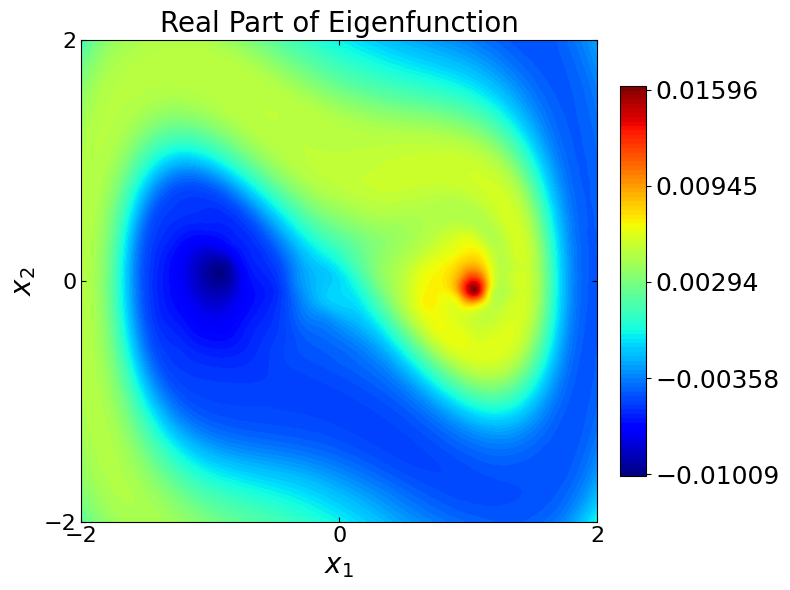}
    \label{fig:pruned_eigenfunc}
  \end{subfigure}
  \caption{Real part of the leading non-trivial Koopman eigenfunction approximations for the system in \eqref{eq:sys1} before (left) and after (right) SPV pruning. The pruned eigenfunction is significantly smoother, demonstrating improved approximation quality achieved by the pruning procedure.}
  \label{fig:eigenfuncs}
\end{figure}

\vspace{-0.15cm}
\section{Conclusions}
In this work, we established a unified framework for subspace pruning to identify approximately Koopman invariant subspaces. Within this framework, we introduced the SPV pruning algorithm, which generalizes existing consistency-based methods. A primary advantage of our approach is its broad applicability to general inner product spaces, extending beyond the standard empirical $L_2$ inner product. Furthermore, to address the computational demands of large-scale, data-driven applications, we developed an efficient rank-one update scheme. This procedure rapidly computes the new principal angles and vectors after a direction is dropped, avoiding the prohibitive cost of full recomputations and ensuring the practical viability of the algorithm.

\section*{Acknowledgments}
The authors would like to thank Masih Haseli for insightful discussions regarding RFB-EDMD and principal angles.

\bibliographystyle{ieeetr}
\bibliography{../bib/alias,../bib/Main-add,../bib/Main,../bib/JC}

\section{Appendix}
We present here a result on principal angles and vectors used for the
proof of Theorem~\ref{thm:efficient_computation}.

\begin{lemma}\longthmtitle{Alternate Characterization of Principal
  Arguments}\label{lemma:cos_theta_max}
  Let $\Uc, \Vc \subset \Hc$ be two subspaces with
  $a = \dim(\Uc) \leq \dim(\Vc) = b$.  Let $\{\theta_j\}_{j=1}^a$ be
  the principal angles between $\Uc$ and $\Vc$, and let
  $\{x_j^{\Uc}\}_{j=1}^a \subset \Uc$ and
  $\{y_j^{\Vc}\}_{j=1}^a \subset \Vc$ be the corresponding principal
  vectors.  Then, for $k \in [a]$,
  \begin{equation}\label{eq:alternate_deftn_pvs} 
    \cos \theta_{a - (k-1)} = \min_{\substack{x \in \Uc \\ x \perp
        x_{a-(k-2)}^{\Uc},\dots,x_a^{\Uc}}} \frac{\| \Pc_{\Vc}(x)
      \|}{\| x \|}  ,
  \end{equation}
  and for a given $k$, the minimizer of \eqref{eq:alternate_deftn_pvs}
  is the principal vector $x_{a-(k-1)}^{\Uc}$. Consequently, the set
  of all minimizers as $k \in [a]$ is exactly the set of principal
  vectors $\{x_j^{\Uc}\}_{j=1}^{a}$.
\end{lemma}

\end{document}